\documentclass{amsart}
\usepackage{graphicx}
\usepackage[ruled,vlined,linesnumbered]{algorithm2e}
\usepackage{xy}

\newtheorem{thm}{Theorem}
\theoremstyle{remark}

\newcommand{\assign}{\leftarrow}
\newcommand{\ZZ}{\mathbf{Z}}
\newcommand{\QQ}{\mathbf{Q}}
\newcommand{\divides}{\mathbin{|}}
\newcommand{\floor}[1]{\lfloor {#1} \rfloor}
\newcommand{\ceil}[1]{\lceil {#1} \rceil}

\DeclareMathOperator{\TFT}{\text{\sc TFT}}
\DeclareMathOperator{\ITFT}{\text{\sc ITFT}}
\DeclareMathOperator{\cfTFT}{\text{\sc CacheFriendlyTFT}}
\DeclareMathOperator{\cfITFT}{\text{\sc CacheFriendlyITFT}}

\newcommand{\cna}{{\scriptstyle a}}
\newcommand{\cnb}{{\scriptstyle b}}
\newcommand{\cnha}{{\scriptstyle {\hat a}}}
\newcommand{\cnhao}{{\scriptstyle (\hspace{-1pt}{\hat a}\hspace{-1pt})}}
\newcommand{\cnz}{{\scriptstyle \cdot}}
\newcommand{\cns}{{\scriptscriptstyle ?}}

\newcommand{\cba}{{\scriptstyle \boldsymbol a}}
\newcommand{\cbb}{{\scriptstyle \boldsymbol b}}
\newcommand{\cbbo}{{\scriptstyle (\hspace{-1pt}\boldsymbol b\hspace{-1pt})}}
\newcommand{\cbha}{{\scriptstyle \boldsymbol{\hat a}}}

\newcommand{\cbz}{{\scriptstyle \boldsymbol \cdot}}

\newcommand{\grid}[1]{\save%
(4, -1.3)*{\text{(#1)}};
(0,0); (8,0) **\dir{-};
(0,1); (8,1) **\dir{-};
(0,2); (8,2) **\dir{-};
(0,3); (8,3) **\dir{-};
(0,4); (8,4) **\dir{-};
(0,5); (8,5) **\dir{-};
(0,6); (8,6) **\dir{-};
(0,7); (8,7) **\dir{-};
(0,8); (8,8) **\dir{-};
(0,0); (0,8) **\dir{-};
(1,0); (1,8) **\dir{-};
(2,0); (2,8) **\dir{-};
(3,0); (3,8) **\dir{-};
(4,0); (4,8) **\dir{-};
(5,0); (5,8) **\dir{-};
(6,0); (6,8) **\dir{-};
(7,0); (7,8) **\dir{-};
(8,0); (8,8) **\dir{-};\restore
}

\newcommand{\row}[9]{\save%
(0.5, #1.58)*{\vphantom{bq} #2};
(1.5, #1.58)*{\vphantom{bq} #3};
(2.5, #1.58)*{\vphantom{bq} #4};
(3.5, #1.58)*{\vphantom{bq} #5};
(4.5, #1.58)*{\vphantom{bq} #6};
(5.5, #1.58)*{\vphantom{bq} #7};
(6.5, #1.58)*{\vphantom{bq} #8};
(7.5, #1.58)*{\vphantom{bq} #9};
\restore
}

\SetKwInOut{Output}{Output}
\SetKw{KwAnd}{and}

\begin{document}

\title{A cache-friendly truncated FFT}
\author{David Harvey}
\begin{abstract}
We describe a cache-friendly version of van der Hoeven's truncated FFT and inverse truncated FFT, focusing on the case of `large' coefficients, such as those arising in the Sch\"onhage--Strassen algorithm for multiplication in $\ZZ[x]$. We describe two implementations and examine their performance.
\end{abstract}

\maketitle

\section{Introduction}
\label{sec:intro}

In typical implementations of the FFT method for dense univariate polynomial multiplication, the input polynomials are zero-padded up to an appropriate power-of-two length, causing a jump in the running time when the lengths cross a power-of-two boundary. Van der Hoeven recently described a multiplication algorithm that greatly reduces the size of these jumps, by introducing a novel TFT (truncated FFT) and ITFT (inverse truncated FFT), achieving relatively smooth performance without sacrificing the simplicity of a power-of-two transform length \cite{vdh1, vdh2}.

However, the transforms that he describes suffer from suboptimal locality. The transforms follow the divide-and-conquer FFT paradigm, recursively splitting the problem into two half-sized transforms. If the transform length is $2^\ell$, and only $2^k$ coefficients fit into a given level of cache, then only the deepest $k$ layers of the transform take advantage of that cache; the remaining $\ell - k$ layers do not.

In this paper we address this difficulty, achieving superior locality by reordering the sequence of butterfly operations in van der Hoeven's transforms. Our strategy is similar to Bailey's algorithm \cite{bailey}. Bailey rearranges the data into a $2^{\ell_1} \times 2^{\ell_2}$ matrix, where $\ell_1 + \ell_2 = \ell$, and then rewrites the transform as $2^{\ell_2}$ column transforms of length $2^{\ell_1}$ followed by $2^{\ell_1}$ row transforms of length $2^{\ell_2}$. The divide-and-conquer algorithm may be regarded as the special case where $\ell_1 = 1$ and $\ell_2 = \ell - 1$. However, when $\ell_i \approx \ell/2$, the working set for each row and column is only about $2^{\ell/2}$ coefficients, greatly improving the algorithm's locality. This method can of course be applied recursively, until the working set for each subtransform fits into the lowest level of cache, making efficient use of the entire memory hierarchy.

It is straightforward to adapt this idea to the TFT, obtaining a decomposition of the TFT into TFTs of half the depth (\S\ref{sec:TFT}). The corresponding decomposition of the ITFT is more involved; it becomes necessary to alternate between ITFTs on the rows and columns in a slightly complicated way (\S\ref{sec:ITFT}).

In \S\ref{sec:applications} we discuss the performance of two implementations. The first is an implementation of the Sch\"onhage--Strassen algorithm \cite{schonhage-strassen} for multiplication in $\ZZ[x]$. The second is an implementation of the Sch\"onhage--Nussbaumer convolution algorithm \cite{schonhage, nussbaumer} for the case of $(\ZZ/m\ZZ)[x]$ where $m$ is an odd word-sized modulus. In both cases the Fourier coefficients occupy relatively large blocks of memory. A natural question is whether the new algorithms are suitable for the more conventional case of `small' coefficients, such as double-precision real or complex coefficients. We offer some speculation in \S\ref{sec:small}, although we have not attempted an implementation.

\section{Notation and setup}

Let $R$ be a commutative ring in which $2$ is invertible. We assume that $R$ contains a principal $M$-th root of unity $\omega$, where $M = 2^m$ for some integer $m \geq 1$; this means that $\omega^M = 1$ and moreover that $\sum_{i=0}^{M-1} \omega^{ij} = 0$ for all $0 < j < M$. We have in mind examples like $R = \ZZ/(2^{M/2} + 1)\ZZ$ and $\omega = 2$, which appears in the Sch\"onhage--Strassen algorithm for multiplication in $\ZZ[x]$.

If $L \divides M$, we denote by $\omega_L$ the principal $L$-th root of unity $\omega^{M/L}$; we then have the compatibility relation $(\omega_{L'})^{L'/L} = \omega_L$ for any $L \divides L' \divides M$.

Now suppose that $L \divides M$, $L = 2^\ell$, and let $\zeta \in R^\times$. Let $(a_0, \ldots, a_{L-1}) \in R^L$. The (weighted) discrete Fourier transform (DFT) is defined by
\begin{equation}
\label{eq:dft}
 \hat a_j = \zeta^{j'} \sum_{i=0}^{L-1} \omega_L^{ij'} a_i,  \quad 0 \leq j < L,
\end{equation}
where $j'$ denotes the length-$\ell$ bit-reversal of $j$.

We define the \emph{truncated Fourier transform} (TFT) as follows. Let $1 \leq z \leq L$ and $1 \leq n \leq L$, and suppose that $a_z = \cdots = a_{L-1} = 0$. Then
 \[ \TFT(L, \zeta, z, n; (a_0, \ldots, a_{z-1})) := (\hat a_0, \ldots, \hat a_{n-1}). \]
In other words, the TFT computes a prescribed initial segment of the transform, assuming that some prescribed final segment of the untransformed data is zero (see Figure \ref{fig:TFT}).

\begin{figure}[ht]
\[ \begin{xy}
(0,0);<10mm,0mm>:<0mm,5mm>::
%
(0,4); (12,4) **\dir{-};
(0,5); (12,5) **\dir{-};
(0,4); (0,5) **\dir{-};
(1,4); (1,5) **\dir{-};
(2,4); (2,5) **\dir{-};
(3,4); (3,5) **\dir{-};
(4,4); (4,5) **\dir{-};
(5,4); (5,5) **\dir{-};
(6,4); (6,5) **\dir{-};
(7,4); (7,5) **\dir{-};
(8,4); (8,5) **\dir{-};
(9,4); (9,5) **\dir{-};
(10,4); (10,5) **\dir{-};
(11,4); (11,5) **\dir{-};
(12,4); (12,5) **\dir{-};
%
(0.5, 4.5) *{\vphantom{bq} a_0};
(1.5, 4.5) *{\vphantom{bq} a_1};
(2.5, 4.5) *{\cdots};
(3.5, 4.5) *{\cdots};
(4.5, 4.5) *{\cdots};
(5.5, 4.5) *{\vphantom{bq} a_{z-1}};
(6.5, 4.5) *{\vphantom{bq} 0};
(7.5, 4.5) *{\cdots};
(8.5, 4.5) *{\cdots};
(9.5, 4.5) *{\cdots};
(10.5, 4.5) *{\cdots};
(11.5, 4.5) *{\vphantom{bq} 0};
%
(0,0); (12,0) **\dir{-};
(0,1); (12,1) **\dir{-};
(0,0); (0,1) **\dir{-};
(1,0); (1,1) **\dir{-};
(2,0); (2,1) **\dir{-};
(3,0); (3,1) **\dir{-};
(4,0); (4,1) **\dir{-};
(5,0); (5,1) **\dir{-};
(6,0); (6,1) **\dir{-};
(7,0); (7,1) **\dir{-};
(8,0); (8,1) **\dir{-};
(9,0); (9,1) **\dir{-};
(10,0); (10,1) **\dir{-};
(11,0); (11,1) **\dir{-};
(12,0); (12,1) **\dir{-};
%
(0.5, 0.5) *{\vphantom{bq} \hat a_0};
(1.5, 0.5) *{\vphantom{bq} \hat a_1};
(2.5, 0.5) *{\cdots};
(3.5, 0.5) *{\cdots};
(4.5, 0.5) *{\cdots};
(5.5, 0.5) *{\cdots};
(6.5, 0.5) *{\cdots};
(7.5, 0.5) *{\vphantom{bq} \hat a_{n-1}};
(8.5, 0.5) *{\vphantom{bq} \hat a_n};
(9.5, 0.5) *{\cdots};
(10.5, 0.5) *{\cdots};
(11.5, 0.5) *{\vphantom{bq} \hat a_{L-1}};
%
(3.5, 2); (3.5, 3) **\dir{-}; (5.5, 3) **\dir{-}; (5.5, 2) **\dir{-}; (3.5, 2) **\dir{-};
(4.5, 2.5) *{\text{TFT}};
%
(0, 3.8); (0, 3.6) **\dir{-}; (6, 3.6) **\dir{-}; (6, 3.8) **\dir{-};
(4.5, 3.6); (4.5, 3) **\dir{-}*\dir{>};
%
(0, 1.2); (0, 1.4) **\dir{-}; (8, 1.4) **\dir{-}; (8, 1.2) **\dir{-};
(4.5, 2); (4.5, 1.4) **\dir{-}*\dir{>};
\end{xy} \]
\caption{The TFT.}
\label{fig:TFT}
\end{figure}

The definition of the \emph{inverse truncated Fourier transform} (ITFT) is more involved. Let $f \in \{0, 1\}$. Suppose that $1 \leq z \leq L$ and $1 \leq n + f \leq L$, and moreover that $z \geq n$. Suppose as before that $a_z = \cdots = a_{L-1} = 0$. Then
\begin{multline*}
 \ITFT(L, \zeta, z, n, f; (\hat a_0, \ldots, \hat a_{n-1}, La_n, \ldots, La_{z-1})) \\
  :=
  \begin{cases}
  (La_0, \ldots, La_{n-1}) & f = 0, \\
  (La_0, \ldots, La_{n-1}, \hat a_n) & f = 1.
  \end{cases}
\end{multline*}
In other words, the ITFT takes as input an initial segment of the transformed data together with the \emph{complementary} final segment of the untransformed data (some components of which are known to be zero), and returns the initial segment of the untransformed data, and optionally (if $f = 1$) the next transformed coordinate (see Figure \ref{fig:ITFT}). When $z = n = L$, $f = 0$ and $\zeta = 1$, the TFT and ITFT reduce to the usual DFT and inverse DFT, with inputs in normal order and outputs in bit-reversed order.
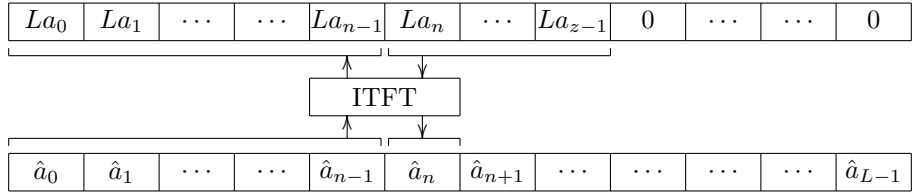
\begin{figure}[ht]
\[ \begin{xy}
(0,0);<10mm,0mm>:<0mm,5mm>::
%
(0,4); (12,4) **\dir{-};
(0,5); (12,5) **\dir{-};
(0,4); (0,5) **\dir{-};
(1,4); (1,5) **\dir{-};
(2,4); (2,5) **\dir{-};
(3,4); (3,5) **\dir{-};
(4,4); (4,5) **\dir{-};
(5,4); (5,5) **\dir{-};
(6,4); (6,5) **\dir{-};
(7,4); (7,5) **\dir{-};
(8,4); (8,5) **\dir{-};
(9,4); (9,5) **\dir{-};
(10,4); (10,5) **\dir{-};
(11,4); (11,5) **\dir{-};
(12,4); (12,5) **\dir{-};
%
(0.5, 4.5) *{\vphantom{bq} La_0};
(1.5, 4.5) *{\vphantom{bq} La_1};
(2.5, 4.5) *{\cdots};
(3.5, 4.5) *{\cdots};
(4.5, 4.5) *{\vphantom{bq} La_{n-1}};
(5.5, 4.5) *{\vphantom{bq} La_{n}};
(6.5, 4.5) *{\cdots};
(7.5, 4.5) *{\vphantom{bq} La_{z-1}};
(8.5, 4.5) *{\vphantom{bq} 0};
(9.5, 4.5) *{\cdots};
(10.5, 4.5) *{\cdots};
(11.5, 4.5) *{\vphantom{bq} 0};
%
(0,0); (12,0) **\dir{-};
(0,1); (12,1) **\dir{-};
(0,0); (0,1) **\dir{-};
(1,0); (1,1) **\dir{-};
(2,0); (2,1) **\dir{-};
(3,0); (3,1) **\dir{-};
(4,0); (4,1) **\dir{-};
(5,0); (5,1) **\dir{-};
(6,0); (6,1) **\dir{-};
(7,0); (7,1) **\dir{-};
(8,0); (8,1) **\dir{-};
(9,0); (9,1) **\dir{-};
(10,0); (10,1) **\dir{-};
(11,0); (11,1) **\dir{-};
(12,0); (12,1) **\dir{-};
%
(0.5, 0.5) *{\vphantom{bq} \hat a_0};
(1.5, 0.5) *{\vphantom{bq} \hat a_1};
(2.5, 0.5) *{\cdots};
(3.5, 0.5) *{\cdots};
(4.5, 0.5) *{\vphantom{bq} \hat a_{n-1}};
(5.5, 0.5) *{\vphantom{bq} \hat a_n};
(6.5, 0.5) *{\vphantom{bq} \hat a_{n+1}};
(7.5, 0.5) *{\cdots};
(8.5, 0.5) *{\cdots};
(9.5, 0.5) *{\cdots};
(10.5, 0.5) *{\cdots};
(11.5, 0.5) *{\vphantom{bq} \hat a_{L-1}};
%
(4, 2); (4, 3) **\dir{-}; (6, 3) **\dir{-}; (6, 2) **\dir{-}; (4, 2) **\dir{-};
(5, 2.5) *{\text{ITFT}};
%
(0, 3.8); (0, 3.6) **\dir{-}; (4.95, 3.6) **\dir{-}; (4.95, 3.8) **\dir{-};
(4.5, 3); (4.5, 3.6) **\dir{-}*\dir{>};
(5.05, 3.8); (5.05, 3.6) **\dir{-}; (8, 3.6) **\dir{-}; (8, 3.8) **\dir{-};
(5.5, 3.6); (5.5, 3) **\dir{-}*\dir{>};
(0, 1.2); (0, 1.4) **\dir{-}; (4.95, 1.4) **\dir{-}; (4.95, 1.2) **\dir{-};
(4.5, 1.4); (4.5, 2) **\dir{-}*\dir{>};
(5.05, 1.2); (5.05, 1.4) **\dir{-}; (6, 1.4) **\dir{-}; (6, 1.2) **\dir{-};
(5.5, 2); (5.5, 1.4) **\dir{-}*\dir{>};
\end{xy} \]
\caption{The ITFT.}
\label{fig:ITFT}
\end{figure}

It is not obvious \emph{a priori} that the ITFT is well-defined, and in particular that the coordinates $\hat a_0, \ldots, \hat a_{n-1}, a_n, \ldots, a_{L-1}$ are linearly independent. Van der Hoeven deduced this from the correctness of his algorithm for computing the ITFT; it will follow in the same way from the proof of correctness of our cache-friendly ITFT algorithm in \S\ref{sec:ITFT}.

Van der Hoeven allowed the input and output coordinates to come from a wider class of subsets of $\{0, \ldots, L-1\}$. In this paper we restrict ourselves to the initial and final segments mentioned above, which suffices for our intended application to univariate polynomial multiplication.

The TFT and ITFT may be used to deduce a polynomial multiplication algorithm in $R[X]$ as follows. Suppose that $g, h \in R[X]$, and let $u = gh$. Let $z_1 = 1 + \deg g$, $z_2 = 1 + \deg h$, $n = z_1 + z_2 - 1$, and assume that $n \leq L$. Let $g_0, \ldots, g_{z_1-1}$ be the coefficients of $g$ and $h_0, \ldots, h_{z_2-1}$ be the coefficients of $h$. Compute
\[ \begin{aligned}
 (\hat g_0, \ldots, \hat g_{n-1}) & = \TFT(L, 1, z_1, n; (g_0, \ldots, g_{z_1-1})), \\
 (\hat h_0, \ldots, \hat h_{n-1}) & = \TFT(L, 1, z_2, n; (h_0, \ldots, h_{z_2-1})),
\end{aligned} \]
and then compute $\hat u_i = \hat g_i \hat h_i$ in $R$ for $0 \leq i < n$. Then $\hat u_0, \ldots, \hat u_{n-1}$ are the first $n$ Fourier coefficients of $u$, and moreover $u_n = \cdots = u_{L-1} = 0$ since $n = \deg u + 1$. Therefore we recover $u$ via
 \[ (Lu_0, \ldots, Lu_{n-1}) = \ITFT(L, 1, n, n, 0; (\hat u_0, \ldots, \hat u_{n-1})). \]
(This multiplication algorithm has not used the parameters $f$ or $\zeta$ in a nontrivial way; these enter the picture when the algorithms are called recursively in \S\ref{sec:TFT} and \S\ref{sec:ITFT}.)

The standard FFT algorithms compute the DFT (or inverse DFT) using $\ell L/2$ `butterfly operations'. In contrast, van der Hoeven showed that the TFT and ITFT may be computed using at most $\ell n/2 + L$ butterfly operations, and we will see that this estimate holds for our cache-friendly TFT and ITFT algorithms as well. Furthermore, in the multiplication algorithm sketched above, only $n$ pointwise multiplications are performed, compared to the $L$ multiplications incurred by the standard FFT method. Therefore, in this simplified algebraic complexity model, the ratio of the running time of the TFT/ITFT-based multiplication algorithm to the running time of the usual FFT multiplication algorithm is $n/L + O(\ell^{-1})$, indicating that the performance is relatively smooth as a function of $n$.

Algorithms \ref{algo:TFT} and \ref{algo:ITFT} below ($\cfTFT$ and $\cfITFT$) implement the TFT and ITFT in a cache-friendly manner. They operate on an array $x_0, \ldots, x_{L-1}$, where $L = 2^\ell$. In general all $L$ elements of the array, even those elements not containing input or output, are used in intermediate computations.

For the TFT, the first $z$ elements are expected to contain the inputs $a_0, \ldots, a_{z-1}$, and the outputs $\hat a_0, \ldots, \hat a_{n-1}$ are written in-place to the same array. For the ITFT, the first $z$ elements are expected to contain the inputs $\hat a_0, \ldots, \hat a_{n-1}, L a_n, \ldots, L a_{z-1}$, and the outputs $L a_0, \ldots, L a_{n-1}$ (optionally followed by $\hat a_n$ if $f = 1$) are written in-place to the same array.

Both algorithms make use of the following well-known decomposition of \eqref{eq:dft}. Let $L = L_1 L_2$ where $L_1 = 2^{\ell_1}$ and $L_2 = 2^{\ell_2}$ (so that $\ell_1 + \ell_2 = \ell$). Write $i = i_2 + L_2 i_1$ where $0 \leq i_1 < L_1$ and $0 \leq i_2 < L_2$, and similarly for $j$. Then $j' = j_1' + L_1 j_2'$, where $j_1'$ and $j_2'$ are respectively the length-$\ell_1$ and length-$\ell_2$ bit-reversals of $j_1$ and $j_2$. We obtain
\[ \begin{aligned}
 \hat a_j = \hat a_{j_2 + L_2 j_1}
  & = \zeta^{j_1' + L_1 j_2'} \sum_{i_2=0}^{L_2-1} \sum_{i_1=0}^{L_1-1} \omega_L^{(i_2 + L_2 i_1)(j_1' + L_1 j_2')} a_{i_2 + L_2 i_1} \\
  & = (\zeta^{L_1})^{j_2'} \sum_{i_2=0}^{L_2-1} \omega_{L_2}^{i_2 j_2'} \left( (\zeta \omega_L^{i_2})^{j_1'} \sum_{i_1=0}^{L_1-1} \omega_{L_1}^{i_1 j_1'} a_{i_2 + L_2 i_1} \right). 
\end{aligned} \]
Therefore if we put
\begin{equation}
\label{eq:b-cols}
 b_k = b_{k_2 + L_2 k_1} = (\zeta \omega_L^{k_2})^{k_1'} \sum_{m=0}^{L_1-1} \omega_{L_1}^{m k_1'} a_{k_2 + L_2 m},
\end{equation}
we obtain
\begin{equation}
\label{eq:b-rows}
 \hat a_j = (\zeta^{L_1})^{j_2'} \sum_{r=0}^{L_2-1} \omega_{L_2}^{r j_2'} b_{r + L_2 j_1}.
\end{equation}
In other words, if $a$, $b$ and $\hat a$ are thought of as $L_1 \times L_2$ matrices, then $b$ is the result of applying an appropriately weighted DFT to each of the columns of $a$, and $\hat a$ is the result of applying an appropriately weighted DFT to each of the rows of $b$.

For the base case $L = 2$ the routines compute the TFT/ITFT directly. If $L = 2^\ell \geq 4$, they write $L = L_1 L_2$ where $L_1 = 2^{\floor{\ell/2}}$ and $L_2 = 2^{\ceil{\ell/2}}$, so that $1 < L_1 < L$ and $1 < L_2 < L$. They treat the array as an $L_1 \times L_2$ matrix, and recurse into TFTs/ITFTs on the columns and rows. The column transforms correspond to recursively applying the TFT/ITFT to the transform given by \eqref{eq:b-cols}; the row transforms similarly correspond to the transform given by \eqref{eq:b-rows}. (Van der Hoeven's TFT and ITFT algorithms are essentially the special case obtained by taking $L_1 = 2$ and $L_2 = L/2$.)

We will denote by $c_u$ the $u$-th column $(x_u, x_{u+L_2}, \ldots, x_{u+(L_1-1)L_2})$ and by $r_u$ the $u$-th row $(x_{uL_2}, x_{uL_2+1}, \ldots, x_{uL_2+L_2-1})$. A real implementation would use auxiliary variables to describe such sub-arrays; for example, a pointer to the first element and a stride parameter.

Common to both routines is the decomposition $n = n_2 + L_2 n_1$ where $0 \leq n_1 \leq L_1$ and $0 \leq n_2 < L_2$, and where $n_1 = L_1$ implies $n_2 = 0$. This partitions the first $n$ cells of the array into $n_1$ complete rows followed by $n_2$ cells in the subsequent row. The parameter $z$ is decomposed similarly into $z_1$ and $z_2$.

\section{A cache-friendly TFT}
\label{sec:TFT}

We first consider the TFT; the idea is to compute only those parts of the DFT that are requested. We handle the column transforms first, followed by the row transforms.

\begin{algorithm}
\nllabel{algo:TFT}
\SetLine
\dontprintsemicolon
\KwIn{$L = 2^\ell \geq 2$, $\zeta \in R^\times$, \newline $1 \leq z \leq L$, $1 \leq n \leq L$, \newline $x_i = a_i$ for $0 \leq i < z$}\;
\KwOut{$x_i = \hat a_i$ for $0 \leq i < n$}\;
\medskip
\If{$L = 2$}{
  \tcp{base case}\;
  \lIf{$n = 2$ \KwAnd $z = 2$}{$(x_0, x_1) \assign (x_0 + x_1, \zeta(x_0 - x_1))$}\nllabel{line:TFT-n2z2}\;
  \lIf{$n = 2$ \KwAnd $z = 1$}{$x_1 \assign \zeta x_0$}\nllabel{line:TFT-n2z1}\;
  \lIf{$n = 1$ \KwAnd $z = 2$}{$x_0 \assign x_0 + x_1$}\nllabel{line:TFT-n1z2}\;
  \KwRet\;
}
\medskip
\tcp{recursive case}\;
$L_1 \assign 2^{\floor{\ell/2}}$, $L_2 \assign 2^{\ceil{\ell/2}}$\nllabel{line:TFT-split}\;
$n_2 \assign n \bmod L_2$, $n_1 \assign \floor{n / L_2}$, $n_1' \assign \ceil{n / L_2}$\nllabel{line:TFT-n-partition}\;
$z_2 \assign z \bmod L_2$, $z_1 \assign \floor{z / L_2}$\nllabel{line:TFT-z-partition}\;
\lIf{$z_1 > 0$}{$z_2' \assign L_2$} \lElse{$z_2' \assign z_2$}\nllabel{line:TFT-z2'}\;
\medskip
\tcp{column transforms}
\lFor{$0 \leq u < z_2$}{$\cfTFT(L_1, \omega_L^u \zeta, z_1 + 1, n_1'; c_u)$}\nllabel{line:TFT-begin-cols}\;
\lFor{$z_2 \leq u < z_2'$}{$\cfTFT(L_1, \omega_L^u \zeta, z_1, n_1'; c_u)$}\nllabel{line:TFT-end-cols}\;
\medskip
\tcp{row transforms}
\lFor{$0 \leq u < n_1$}{$\cfTFT(L_2, \zeta^{L_1}, z_2', L_2; r_u)$}\nllabel{line:TFT-begin-rows}\;
\lIf{$n_2 > 0$}{$\cfTFT(L_2, \zeta^{L_1}, z_2', n_2; r_{n_1})$}\nllabel{line:TFT-end-rows}\;
\caption{$\cfTFT(L, \zeta, z, n; (x_0, \ldots, x_{L-1}))$}
\end{algorithm}

\begin{thm}
\label{thm:TFT}
Algorithm \ref{algo:TFT} correctly computes the TFT. The base case is executed at most $\min((n-1)\ell/2 + L - 1, L \ell/2)$ times.
\end{thm}
\begin{proof}
We first consider the base case $L = 2$. The relevant DFT is given by $(\hat a_0, \hat a_1) = (a_0 + a_1, \zeta(a_0 - a_1))$. If $z = 1$ then $a_1 = 0$, and the transform becomes simply $(\hat a_0, \hat a_1) = (a_0, \zeta a_0)$. If $n = 2$ then both $\hat a_0$ and $\hat a_1$ must be computed; if $n = 1$ then only $\hat a_0$ is needed. Lines \ref{line:TFT-n2z2}--\ref{line:TFT-n1z2} handle the various cases.

\begin{figure}[ht]
\[ 
\begin{xy} (0,0);<3mm,0mm>:<0mm,3mm>::
{\grid{a}};
{\row7  \cna  \cna  \cna  \cna  \cna  \cnz  \cnz  \cnz  };
{\row6  \cnz  \cnz  \cnz  \cnz  \cnz  \cnz  \cnz  \cnz  };
{\row5  \cnz  \cnz  \cnz  \cnz  \cnz  \cnz  \cnz  \cnz  };
{\row4  \cnz  \cnz  \cnz  \cnz  \cnz  \cnz  \cnz  \cnz  };
{\row3  \cnz  \cnz  \cnz  \cnz  \cnz  \cnz  \cnz  \cnz  };
{\row2  \cnz  \cnz  \cnz  \cnz  \cnz  \cnz  \cnz  \cnz  };
{\row1  \cnz  \cnz  \cnz  \cnz  \cnz  \cnz  \cnz  \cnz  };
{\row0  \cnz  \cnz  \cnz  \cnz  \cnz  \cnz  \cnz  \cnz  };
(0, 8.2); (0, 8.4) **\dir{-}; (5, 8.4) **\dir{-}; (5, 8.2); **\dir{-};
(2.5, 9) *{z_2};
\end{xy}
\quad
\begin{xy} (0,0);<3mm,0mm>:<0mm,3mm>::
{\grid{b}};
{\row7  \cna  \cna  \cna  \cna  \cna  \cna  \cna  \cna};
{\row6  \cna  \cna  \cna  \cna  \cna  \cna  \cna  \cna};
{\row5  \cna  \cna  \cna  \cna  \cna  \cna  \cna  \cna};
{\row4  \cna  \cna  \cna  \cna  \cna  \cna  \cna  \cna};
{\row3  \cna  \cna  \cna  \cna  \cna  \cna  \cna  \cna};
{\row2  \cnz  \cnz  \cnz  \cnz  \cnz  \cnz  \cnz  \cnz};
{\row1  \cnz  \cnz  \cnz  \cnz  \cnz  \cnz  \cnz  \cnz};
{\row0  \cnz  \cnz  \cnz  \cnz  \cnz  \cnz  \cnz  \cnz};
(-0.2, 3); (-0.4, 3) **\dir{-}; (-0.4, 8) **\dir{-}; (-0.2, 8) **\dir{-}; (-1, 5.5) *{z_1};
\end{xy}
\quad
\begin{xy} (0,0);<3mm,0mm>:<0mm,3mm>::
{\grid{c}};
{\row7  \cna  \cna  \cna  \cna  \cna  \cna  \cna  \cna};
{\row6  \cna  \cna  \cna  \cna  \cna  \cna  \cna  \cna};
{\row5  \cna  \cna  \cna  \cna  \cna  \cna  \cna  \cna};
{\row4  \cna  \cna  \cna  \cna  \cna  \cna  \cna  \cna};
{\row3  \cna  \cna  \cna  \cna  \cna  \cna  \cna  \cna};
{\row2  \cna  \cna  \cna  \cna  \cna  \cnz  \cnz  \cnz};
{\row1  \cnz  \cnz  \cnz  \cnz  \cnz  \cnz  \cnz  \cnz};
{\row0  \cnz  \cnz  \cnz  \cnz  \cnz  \cnz  \cnz  \cnz};
(0, 8.2); (0, 8.4) **\dir{-}; (5, 8.4) **\dir{-}; (5, 8.2); **\dir{-};
(2.5, 9) *{z_2};
(-0.2, 3); (-0.4, 3) **\dir{-}; (-0.4, 8) **\dir{-}; (-0.2, 8) **\dir{-}; (-1, 5.5) *{z_1};
\end{xy}
\]
\caption{Before line \ref{line:TFT-begin-cols} of $\cfTFT$.}
\label{fig:TFT-input}
\end{figure}

\begin{figure}[ht]
\[ 
\begin{xy} (0,0);<3mm,0mm>:<0mm,3mm>::
{\grid{a}};
{\row7  \cnb  \cnb  \cnb  \cnb  \cnb  \cnz  \cnz  \cnz};
{\row6  \cnb  \cnb  \cnb  \cnb  \cnb  \cnz  \cnz  \cnz};
{\row5  \cnb  \cnb  \cnb  \cnb  \cnb  \cnz  \cnz  \cnz};
{\row4  \cnb  \cnb  \cnb  \cnb  \cnb  \cnz  \cnz  \cnz};
{\row3  \cnb  \cnb  \cnb  \cnb  \cnb  \cnz  \cnz  \cnz};
{\row2  \cnb  \cnb  \cnb  \cnb  \cnb  \cnz  \cnz  \cnz};
{\row1  \cns  \cns  \cns  \cns  \cns  \cnz  \cnz  \cnz};
{\row0  \cns  \cns  \cns  \cns  \cns  \cnz  \cnz  \cnz};
(0, 8.2); (0, 8.4) **\dir{-}; (5, 8.4) **\dir{-}; (5, 8.2); **\dir{-};
(2.5, 9.3) *{z_2'};
(-0.2, 2); (-0.4, 2) **\dir{-}; (-0.4, 8) **\dir{-}; (-0.2, 8) **\dir{-}; (-1.1, 5) *{n_1'};
\end{xy}
\quad
\begin{xy} (0,0);<3mm,0mm>:<0mm,3mm>::
{\grid{b}};
{\row7  \cnb  \cnb  \cnb  \cnb  \cnb  \cnb  \cnb  \cnb};
{\row6  \cnb  \cnb  \cnb  \cnb  \cnb  \cnb  \cnb  \cnb};
{\row5  \cnb  \cnb  \cnb  \cnb  \cnb  \cnb  \cnb  \cnb};
{\row4  \cnb  \cnb  \cnb  \cnb  \cnb  \cnb  \cnb  \cnb};
{\row3  \cnb  \cnb  \cnb  \cnb  \cnb  \cnb  \cnb  \cnb};
{\row2  \cnb  \cnb  \cnb  \cnb  \cnb  \cnb  \cnb  \cnb};
{\row1  \cns  \cns  \cns  \cns  \cns  \cns  \cns  \cns};
{\row0  \cns  \cns  \cns  \cns  \cns  \cns  \cns  \cns};
(0, 8.2); (0, 8.4) **\dir{-}; (8, 8.4) **\dir{-}; (8, 8.2); **\dir{-};
(4, 9.3) *{z_2'};
(-0.2, 2); (-0.4, 2) **\dir{-}; (-0.4, 8) **\dir{-}; (-0.2, 8) **\dir{-}; (-1.1, 5) *{n_1'};
\end{xy}
\]
\caption{After line \ref{line:TFT-end-cols} of $\cfTFT$.}
\label{fig:TFT-after-cols}
\end{figure}

\begin{figure}[ht]
\[ 
\begin{xy} (0,0);<3mm,0mm>:<0mm,3mm>::
{\grid{a}};
{\row7  \cnha \cnha \cnha \cnha \cnha \cnha \cnha \cnha};
{\row6  \cnha \cnha \cnha \cnha \cnha \cnha \cnha \cnha};
{\row5  \cnha \cnha \cnha \cnha \cnha \cnha \cnha \cnha};
{\row4  \cnha \cnha \cnha \cnha \cnha \cnha \cnha \cnha};
{\row3  \cnha \cnha \cnha \cnha \cnha \cnha \cnha \cnha};
{\row2  \cnha \cnha \cnha \cns  \cns  \cns  \cns  \cns};
{\row1  \cns  \cns  \cns  \cns  \cns  \cns  \cns  \cns};
{\row0  \cns  \cns  \cns  \cns  \cns  \cns  \cns  \cns};
(0, 8.2); (0, 8.4) **\dir{-}; (3, 8.4) **\dir{-}; (3, 8.2); **\dir{-};
(1.5, 9) *{n_2};
(-0.2, 3); (-0.4, 3) **\dir{-}; (-0.4, 8) **\dir{-}; (-0.2, 8) **\dir{-}; (-1.1, 5.5) *{n_1};
\end{xy}
\quad
\begin{xy} (0,0);<3mm,0mm>:<0mm,3mm>::
{\grid{b}};
{\row7  \cnha \cnha \cnha \cnha \cnha \cnha \cnha \cnha};
{\row6  \cnha \cnha \cnha \cnha \cnha \cnha \cnha \cnha};
{\row5  \cnha \cnha \cnha \cnha \cnha \cnha \cnha \cnha};
{\row4  \cnha \cnha \cnha \cnha \cnha \cnha \cnha \cnha};
{\row3  \cnha \cnha \cnha \cnha \cnha \cnha \cnha \cnha};
{\row2  \cns  \cns  \cns  \cns  \cns  \cns  \cns  \cns};
{\row1  \cns  \cns  \cns  \cns  \cns  \cns  \cns  \cns};
{\row0  \cns  \cns  \cns  \cns  \cns  \cns  \cns  \cns};
(-0.2, 3); (-0.4, 3) **\dir{-}; (-0.4, 8) **\dir{-}; (-0.2, 8) **\dir{-}; (-1.1, 5.5) *{n_1};
\end{xy}
\]
\caption{After lines \ref{line:TFT-begin-rows}--\ref{line:TFT-end-rows} of $\cfTFT$.}
\label{fig:TFT-rows}
\end{figure}

Now we consider the recursive case, for $L = 2^\ell \geq 4$. Figures \ref{fig:TFT-input}(a)--(c) show the possible input configurations, for $L = 64$, $L_1 = L_2 = 8$. Cells labelled $a$ contain some $a_i$; cells labelled $\cdot$ contain uninitialised data, but implicitly represent $a_i = 0$. Diagram (a) shows the case $z_1 = 0$, in which case $z_2' = z_2$. Diagram (b) shows the case $z_1 > 0$ and $z_2 = 0$, and diagram (c) shows the case $z_1 > 0$, $z_2 > 0$. In these latter cases $z_2' = L_2$. Lines \ref{line:TFT-begin-cols}--\ref{line:TFT-end-cols} apply the TFT recursively to the columns to evaluate the first $n_1'$ rows of \eqref{eq:b-cols}. Line \ref{line:TFT-begin-cols} handles those columns containing $z_1 + 1$ nonzero entries; line \ref{line:TFT-end-cols} handles those containing only $z_1$ nonzero entries.

After lines \ref{line:TFT-begin-cols}--\ref{line:TFT-end-cols} have executed, we have $x_i = b_i$ for $0 \leq i_1 < n_1'$ and $0 \leq i_2 < z_2'$, and we also know that $b_i = 0$ for $z_2' \leq i < L_2$ (the latter statement is non-vacuous only if $z_1 = 0$). Figure \ref{fig:TFT-after-cols} illustrates the situation: cells labelled $b$ contain some $b_i$; cells labelled $\cdot$ contain unspecified data but implicitly represent $b_i = 0$; cells labelled $?$ are meaningless. Diagram (a) shows the case $z_2' < L_2$, and diagram (b) shows $z_2' = L_2$.

Next, lines \ref{line:TFT-begin-rows}--\ref{line:TFT-end-rows} apply the TFT recursively to the first $n_1'$ rows to evaluate \eqref{eq:b-rows}. Figure \ref{fig:TFT-rows} shows the possible output configurations. Cells labelled $\hat a$ contain some $\hat a_i$; cells labelled $?$ contain meaningless data. Diagram (a) shows the case $n_2 > 0$, where $n_1' = n_1 + 1$, and diagram (b) shows the case $n_2 = 0$, where $n_1' = n_1$. Line \ref{line:TFT-begin-rows} handles the first $n_1$ rows, where $\hat a_i$ must be computed for $0 \leq i_2 < L_2$; line \ref{line:TFT-end-rows} handles the remaining partial row, where $\hat a_i$ is needed only for $0 \leq i_2 < n_2$.

We prove the complexity estimate by induction on $L$. For $L = 2$ the bound is $\min((n-1)/2 + 1, 1) = 1$, so the estimate holds. Now assume that $L \geq 4$, and let $\ell_1 = \log_2 L_1$ and $\ell_2 = \log_2 L_2$.

We first verify that the number of calls to the base case is bounded by $L\ell/2$. By induction, lines \ref{line:TFT-begin-cols}--\ref{line:TFT-end-cols} call the base case at most $L_2(L_1 \ell_1/2)$ times, and lines \ref{line:TFT-begin-rows}--\ref{line:TFT-end-rows} call it at most $n_1'(L_2 \ell_2/2) \leq L_1(L_2 \ell_2/2)$ times. The sum is $L_1 L_2 (\ell_1 + \ell_2)/2 = L\ell/2$.

Second, we must verify that the number of calls is bounded by $(n-1)\ell/2 + L - 1$. Let $\delta = n_1' - n_1 \in \{0, 1\}$. Lines \ref{line:TFT-begin-cols}--\ref{line:TFT-end-cols} call the base case at most $L_2 ((n_1 + \delta - 1) \ell_1/2 + L_1 - 1)$ times, line \ref{line:TFT-begin-rows} calls it at most $n_1(L_2 \ell_2/2)$ times, and line \ref{line:TFT-end-rows} calls it at most $\delta ((n_2 - 1)\ell_2/2 + L_2 - 1)$ times. The sum of these terms is $\frac 12 X + Y$ where
\[ \begin{aligned}
 X & = L_2(n_1 - 1)\ell_1 + n_1 L_2 \ell_2 + \delta(L_2 \ell_1 + (n_2 - 1)\ell_2) \\
   & = (n - n_2)\ell - L_2 \ell_1 + \delta(L_2 \ell_1 + (n_2 - 1)\ell_2) \\
   & = (n - 1)\ell + (\delta - 1)L_2 \ell_1 + (n_2 - 1)(\delta \ell_2 - \ell), \\
 Y & = L_2 (L_1 - 1) + \delta(L_2 - 1) = L - 1 + (\delta - 1)(L_2 - 1).
\end{aligned} \]
If $\delta = 1$, then $n_2 \geq 1$ and $(n_2 - 1)(\delta \ell_2 - \ell) = -\ell_1(n_2 - 1) \leq 0$. If $\delta = 0$ then $n_2 = 0$ and $(\delta - 1)L_2\ell_1 + (n_2 - 1)(\delta \ell_2 - \ell) = -L_2\ell_1 + \ell_1 + \ell_2$, which is non-positive since $L_2 = 2^{\ell_2} \geq \ell_2 + 1$. The desired estimate holds in both cases.
\end{proof}

\section{A cache-friendly inverse TFT}
\label{sec:ITFT}

The ITFT cannot be implemented by simply running the TFT in reverse, because when the ITFT commences there is insufficient information to perform all the row transforms. In particular, if $n \not\equiv 0 \bmod{L_2}$, then the $\floor{n/L_2}$-th row contains some $\hat a_i$ but does not contain the corresponding $b_i$ needed to apply \eqref{eq:b-rows}.

To circumvent this difficulty, we proceed as follows. We first perform as many row transforms as possible. We are then able to perform \emph{some} of the column transforms. When these are complete, it becomes possible to execute the last row transform that was inaccessible before. After this row transform, the remainder of the column transforms may be completed. Algorithm \ref{algo:ITFT} gives a precise statement.

\begin{algorithm}
\nllabel{algo:ITFT}
\SetLine
\dontprintsemicolon
\KwIn{$L = 2^\ell \geq 2$, $\zeta \in R^\times$, \newline $f \in \{0, 1\}$, $1 \leq n + f \leq L$, $1 \leq z \leq L$, $z \geq n$, \newline $x_i = \hat a_i$ for $0 \leq i < n$, $x_i = La_i$ for $n \leq i < z$}\;
\KwOut{$x_i = La_i$ for $0 \leq i < n$, \newline $x_n = \hat a_n$ if $f = 1$}\;
\medskip
\If{$L = 2$}{
  \tcp{base case}\;
  \lIf{$n = 2$}{$(x_0, x_1) \assign (x_0 + \zeta^{-1} x_1, x_0 - \zeta^{-1} x_1)$}\nllabel{line:ITFT-n2}\;
  \lIf{$n = 1$ \KwAnd $f = 1$ \KwAnd $z = 2$}{$(x_0, x_1) \assign (2x_0 - x_1, \zeta(x_0 - x_1))$}\nllabel{line:ITFT-n1f1z2}\;
  \lIf{$n = 1$ \KwAnd $f = 1$ \KwAnd $z = 1$}{$(x_0, x_1) \assign (2x_0, \zeta x_0)$}\nllabel{line:ITFT-n1-begin-other}\;
  \lIf{$n = 1$ \KwAnd $f = 0$ \KwAnd $z = 2$}{$x_0 \assign 2x_0 - x_1$}\;
  \lIf{$n = 1$ \KwAnd $f = 0$ \KwAnd $z = 1$}{$x_0 \assign 2x_0$}\nllabel{line:ITFT-n1-end-other}\;
  \lIf{$n = 0$ \KwAnd $z = 2$}{$x_0 \assign (x_0 + x_1)/2$}\nllabel{line:ITFT-n0z2}\;
  \lIf{$n = 0$ \KwAnd $z = 1$}{$x_0 \assign x_0/2$}\nllabel{line:ITFT-n0z1}\;
  \KwRet\;
}
\medskip
\tcp{recursive case}\;

$L_1 \assign 2^{\floor{\ell/2}}$, $L_2 \assign 2^{\ceil{\ell/2}}$\;
$n_2 \assign n \bmod L_2$, $n_1 \assign \floor{n / L_2}$\;
$z_2 \assign z \bmod L_2$, $z_1 \assign \floor{z / L_2}$\;
\lIf{$n_2 + f > 0$}{$f' \assign 1$} \lElse{$f' \assign 0$}\;
\lIf{$z_1 > 0$}{$z_2' \assign L_2$} \lElse{$z_2' \assign z_2$}\;
$m \assign \min(n_2, z_2)$, $m' \assign \max(n_2, z_2)$\;
\medskip
\tcp{row tranforms}
\lFor{$0 \leq u < n_1$}{$\cfITFT(L_2, \zeta^{L_1}, L_2, L_2, 0; r_u)$}\nllabel{line:ITFT-first-rows}\;
\medskip
\tcp{rightmost column transforms}
\lFor{$n_2 \leq u < m'$}{$\cfITFT(L_1, \omega_L^u \zeta, z_1 + 1, n_1, f'; c_u)$}\nllabel{line:ITFT-first-cols-1}\;
\lFor{$m' \leq u < z_2'$}{$\cfITFT(L_1, \omega_L^u \zeta, z_1, n_1, f'; c_u)$}\nllabel{line:ITFT-first-cols-2}\;
\medskip
\tcp{last row transform}
\lIf{$f' = 1$}{$\cfITFT(L_2, \zeta^{L_1}, z_2', n_2, f; r_{n_1})$}\nllabel{line:ITFT-last-row}\;
\medskip
\tcp{leftmost column transforms}
\lFor{$0 \leq u < m$}{$\cfITFT(L_1, \omega_L^u \zeta, z_1 + 1, n_1 + 1, 0; c_u)$}\nllabel{line:ITFT-last-cols-1}\;
\lFor{$m \leq u < n_2$}{$\cfITFT(L_1, \omega_L^u \zeta, z_1, n_1 + 1, 0; c_u)$}\nllabel{line:ITFT-last-cols-2}\;
\caption{$\cfITFT(L, \zeta, z, n, f; (x_0, \ldots, x_{L-1}))$}
\end{algorithm}

\begin{thm}
\label{thm:ITFT}
Algorithm \ref{algo:ITFT} correctly computes the ITFT. The base case is executed at most $\min((n+f-1)\ell/2 + L - 1, L \ell/2)$ times.
\end{thm}
\begin{proof}
We first consider the base case $L = 2$. As before, the relevant DFT is given by $(\hat a_0, \hat a_1) = (a_0 + a_1, \zeta(a_0 - a_1))$. If $n = 2$, then we must have $z = 2$ and $f = 0$, and we are computing the map $(\hat a_0, \hat a_1) \mapsto (2a_0, 2a_1) = (\hat a_0 + \zeta^{-1} \hat a_1, \hat a_0 - \zeta^{-1} \hat a_1)$. This is handled by line \ref{line:ITFT-n2}. Now suppose that $n = 1$. If $f = 1$ and $z = 2$, we must compute the map $(\hat a_0, 2a_1) \mapsto (2a_0, \hat a_1) = (2\hat a_0 - 2a_1, \zeta(\hat a_0 - 2a_1))$ (van der Hoeven's `cross butterfly'). This is handled by line \ref{line:ITFT-n1f1z2}. Lines \ref{line:ITFT-n1-begin-other}--\ref{line:ITFT-n1-end-other} handle the analogous cases where $f = 0$ (the second output is not needed) or where $z = 1$ ($a_1$ is assumed to be zero). Finally suppose that $n = 0$. Then we must have $f = 1$. If $z = 2$, we must compute $(2a_0, 2a_1) \mapsto \hat a_0 = (2a_0 + 2a_1)/2$. This is handled by line \ref{line:ITFT-n0z2}. The $z = 1$ case (where we assume $a_1 = 0$) is handled by line \ref{line:ITFT-n0z1}.

We now suppose that $L \geq 4$ and consider the four cases below. Figures \ref{fig:ITFT-input}--\ref{fig:ITFT-done} illustrate the various stages of the algorithm for each of these cases. Cells labelled $a$, $b$ and $\hat a$ indicate respectively $La_i$, $L_2 b_i$ or $\hat a_i$; cells labelled $\cdot$ are uninitialised, but implicitly represent $a_i = 0$; cells containing $?$ contain unspecified data not used in subsequent computations. A symbol in parentheses indicates that the symbol is only valid if $f = 1$; if $f = 0$ the cell behaves like a $?$ cell. Cells in bold are those about to be transformed by a recursive call.

\begin{figure}
\[ 
\begin{xy} (0,0);<3mm,0mm>:<0mm,3mm>::
{\grid{a}};
{\row7  \cnha \cnha \cnha \cna  \cna  \cna  \cnz  \cnz};
{\row6  \cnz  \cnz  \cnz  \cnz  \cnz  \cnz  \cnz  \cnz};
{\row5  \cnz  \cnz  \cnz  \cnz  \cnz  \cnz  \cnz  \cnz};
{\row4  \cnz  \cnz  \cnz  \cnz  \cnz  \cnz  \cnz  \cnz};
{\row3  \cnz  \cnz  \cnz  \cnz  \cnz  \cnz  \cnz  \cnz};
{\row2  \cnz  \cnz  \cnz  \cnz  \cnz  \cnz  \cnz  \cnz};
{\row1  \cnz  \cnz  \cnz  \cnz  \cnz  \cnz  \cnz  \cnz};
{\row0  \cnz  \cnz  \cnz  \cnz  \cnz  \cnz  \cnz  \cnz};
(0, 8.2); (0, 8.4) **\dir{-}; (6, 8.4) **\dir{-}; (6, 8.2); **\dir{-}; (3, 9) *{n_2};
(0, -0.2); (0, -0.4) **\dir{-}; (3, -0.4) **\dir{-}; (3, -0.2); **\dir{-}; (0.5, -1) *{z_2};
\end{xy}
\hspace{5pt}
\begin{xy} (0,0);<3mm,0mm>:<0mm,3mm>::
{\grid{b}};
{\row7  \cbha \cbha \cbha \cbha \cbha \cbha \cbha \cbha};
{\row6  \cbha \cbha \cbha \cbha \cbha \cbha \cbha \cbha};
{\row5  \cbha \cbha \cbha \cbha \cbha \cbha \cbha \cbha};
{\row4  \cna  \cna  \cna  \cna  \cna  \cna  \cna  \cna };
{\row3  \cna  \cna  \cna  \cna  \cna  \cna  \cna  \cna };
{\row2  \cna  \cna  \cna  \cna  \cnz  \cnz  \cnz  \cnz };
{\row1  \cnz  \cnz  \cnz  \cnz  \cnz  \cnz  \cnz  \cnz };
{\row0  \cnz  \cnz  \cnz  \cnz  \cnz  \cnz  \cnz  \cnz };
(-0.2, 5); (-0.4, 5) **\dir{-}; (-0.4, 8) **\dir{-}; (-0.2, 8) **\dir{-}; (-1.1, 6.8) *{n_1};
(8.2, 3); (8.4, 3) **\dir{-}; (8.4, 8) **\dir{-}; (8.2, 8) **\dir{-}; (9.1, 5.2) *{z_1};
(0, -0.2); (0, -0.4) **\dir{-}; (4, -0.4) **\dir{-}; (4, -0.2); **\dir{-}; (0.5, -1) *{z_2};
\end{xy}
\hspace{-3pt}
\begin{xy} (0,0);<3mm,0mm>:<0mm,3mm>::
{\grid{c}};
{\row7  \cbha \cbha \cbha \cbha \cbha \cbha \cbha \cbha};
{\row6  \cbha \cbha \cbha \cbha \cbha \cbha \cbha \cbha};
{\row5  \cbha \cbha \cbha \cbha \cbha \cbha \cbha \cbha};
{\row4  \cnha \cnha \cna  \cna  \cna  \cna  \cna  \cna };
{\row3  \cna  \cna  \cna  \cna  \cna  \cna  \cna  \cna };
{\row2  \cna  \cna  \cna  \cna  \cnz  \cnz  \cnz  \cnz };
{\row1  \cnz  \cnz  \cnz  \cnz  \cnz  \cnz  \cnz  \cnz };
{\row0  \cnz  \cnz  \cnz  \cnz  \cnz  \cnz  \cnz  \cnz };
(-0.2, 5); (-0.4, 5) **\dir{-}; (-0.4, 8) **\dir{-}; (-0.2, 8) **\dir{-}; (-1.1, 6.8) *{n_1};
(0, 8.2); (0, 8.4) **\dir{-}; (2, 8.4) **\dir{-}; (2, 8.2); **\dir{-}; (1, 9) *{n_2};
(8.2, 3); (8.4, 3) **\dir{-}; (8.4, 8) **\dir{-}; (8.2, 8) **\dir{-}; (9.1, 5.2) *{z_1};
(0, -0.2); (0, -0.4) **\dir{-}; (4, -0.4) **\dir{-}; (4, -0.2); **\dir{-}; (0.5, -1) *{z_2};
\end{xy}
\hspace{-3pt}
\begin{xy} (0,0);<3mm,0mm>:<0mm,3mm>::
{\grid{d}};
{\row7  \cbha \cbha \cbha \cbha \cbha \cbha \cbha \cbha};
{\row6  \cbha \cbha \cbha \cbha \cbha \cbha \cbha \cbha};
{\row5  \cbha \cbha \cbha \cbha \cbha \cbha \cbha \cbha};
{\row4  \cnha \cnha \cnha \cnha \cnha \cnha \cna  \cna };
{\row3  \cna  \cna  \cna  \cna  \cna  \cna  \cna  \cna };
{\row2  \cna  \cna  \cna  \cna  \cnz  \cnz  \cnz  \cnz };
{\row1  \cnz  \cnz  \cnz  \cnz  \cnz  \cnz  \cnz  \cnz };
{\row0  \cnz  \cnz  \cnz  \cnz  \cnz  \cnz  \cnz  \cnz };
(-0.2, 5); (-0.4, 5) **\dir{-}; (-0.4, 8) **\dir{-}; (-0.2, 8) **\dir{-}; (-1.1, 6.8) *{n_1};
(0, 8.2); (0, 8.4) **\dir{-}; (6, 8.4) **\dir{-}; (6, 8.2); **\dir{-}; (3, 9) *{n_2};
(8.2, 3); (8.4, 3) **\dir{-}; (8.4, 8) **\dir{-}; (8.2, 8) **\dir{-}; (9.1, 5.2) *{z_1};
(0, -0.2); (0, -0.4) **\dir{-}; (4, -0.4) **\dir{-}; (4, -0.2); **\dir{-}; (0.5, -1) *{z_2};
\end{xy}
\]
\caption{Before line \ref{line:ITFT-first-rows} of $\cfITFT$. The bold rows are about to be transformed by line \ref{line:ITFT-first-rows}.}
\label{fig:ITFT-input}
\end{figure}
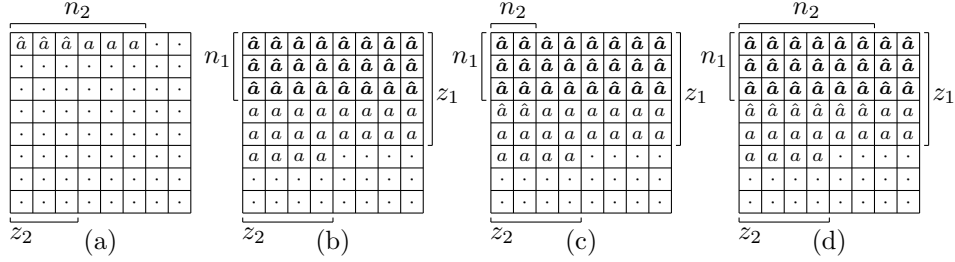

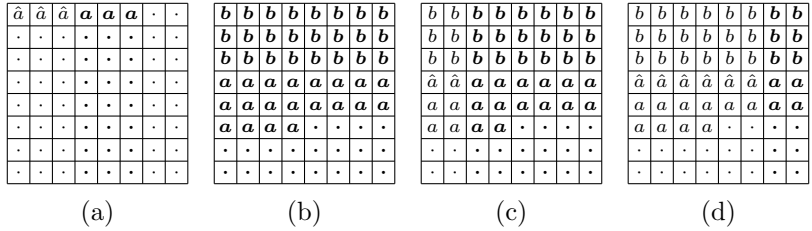
\begin{figure}
\[ 
\begin{xy} (0,0);<3mm,0mm>:<0mm,3mm>::
{\grid{a}};
{\row7  \cnha \cnha \cnha \cba  \cba  \cba  \cnz  \cnz};
{\row6  \cnz  \cnz  \cnz  \cbz  \cbz  \cbz  \cnz  \cnz};
{\row5  \cnz  \cnz  \cnz  \cbz  \cbz  \cbz  \cnz  \cnz};
{\row4  \cnz  \cnz  \cnz  \cbz  \cbz  \cbz  \cnz  \cnz};
{\row3  \cnz  \cnz  \cnz  \cbz  \cbz  \cbz  \cnz  \cnz};
{\row2  \cnz  \cnz  \cnz  \cbz  \cbz  \cbz  \cnz  \cnz};
{\row1  \cnz  \cnz  \cnz  \cbz  \cbz  \cbz  \cnz  \cnz};
{\row0  \cnz  \cnz  \cnz  \cbz  \cbz  \cbz  \cnz  \cnz};
\end{xy}
\quad
\begin{xy} (0,0);<3mm,0mm>:<0mm,3mm>::
{\grid{b}};
{\row7  \cbb  \cbb  \cbb  \cbb  \cbb  \cbb  \cbb  \cbb };
{\row6  \cbb  \cbb  \cbb  \cbb  \cbb  \cbb  \cbb  \cbb };
{\row5  \cbb  \cbb  \cbb  \cbb  \cbb  \cbb  \cbb  \cbb };
{\row4  \cba  \cba  \cba  \cba  \cba  \cba  \cba  \cba };
{\row3  \cba  \cba  \cba  \cba  \cba  \cba  \cba  \cba };
{\row2  \cba  \cba  \cba  \cba  \cbz  \cbz  \cbz  \cbz };
{\row1  \cbz  \cbz  \cbz  \cbz  \cbz  \cbz  \cbz  \cbz };
{\row0  \cbz  \cbz  \cbz  \cbz  \cbz  \cbz  \cbz  \cbz };
\end{xy}
\quad
\begin{xy} (0,0);<3mm,0mm>:<0mm,3mm>::
{\grid{c}};
{\row7  \cnb  \cnb  \cbb  \cbb  \cbb  \cbb  \cbb  \cbb };
{\row6  \cnb  \cnb  \cbb  \cbb  \cbb  \cbb  \cbb  \cbb };
{\row5  \cnb  \cnb  \cbb  \cbb  \cbb  \cbb  \cbb  \cbb };
{\row4  \cnha \cnha \cba  \cba  \cba  \cba  \cba  \cba };
{\row3  \cna  \cna  \cba  \cba  \cba  \cba  \cba  \cba };
{\row2  \cna  \cna  \cba  \cba  \cbz  \cbz  \cbz  \cbz };
{\row1  \cnz  \cnz  \cbz  \cbz  \cbz  \cbz  \cbz  \cbz };
{\row0  \cnz  \cnz  \cbz  \cbz  \cbz  \cbz  \cbz  \cbz };
\end{xy}
\quad
\begin{xy} (0,0);<3mm,0mm>:<0mm,3mm>::
{\grid{d}};
{\row7  \cnb  \cnb  \cnb  \cnb  \cnb  \cnb  \cbb  \cbb };
{\row6  \cnb  \cnb  \cnb  \cnb  \cnb  \cnb  \cbb  \cbb };
{\row5  \cnb  \cnb  \cnb  \cnb  \cnb  \cnb  \cbb  \cbb };
{\row4  \cnha \cnha \cnha \cnha \cnha \cnha \cba  \cba };
{\row3  \cna  \cna  \cna  \cna  \cna  \cna  \cba  \cba };
{\row2  \cna  \cna  \cna  \cna  \cnz  \cnz  \cbz  \cbz };
{\row1  \cnz  \cnz  \cnz  \cnz  \cnz  \cnz  \cbz  \cbz };
{\row0  \cnz  \cnz  \cnz  \cnz  \cnz  \cnz  \cbz  \cbz };
\end{xy}
\]
\caption{After line \ref{line:ITFT-first-rows} of $\cfITFT$. The bold columns are about to be transformed by lines \ref{line:ITFT-first-cols-1}--\ref{line:ITFT-first-cols-2}.}
\label{fig:ITFT-after-rows}
\end{figure}

\begin{figure}
\[ 
\begin{xy} (0,0);<3mm,0mm>:<0mm,3mm>::
{\grid{a}};
{\row7  \cbha \cbha \cbha \cbb  \cbb  \cbb  \cbz  \cbz};
{\row6  \cnz  \cnz  \cnz  \cns  \cns  \cns  \cnz  \cnz};
{\row5  \cnz  \cnz  \cnz  \cns  \cns  \cns  \cnz  \cnz};
{\row4  \cnz  \cnz  \cnz  \cns  \cns  \cns  \cnz  \cnz};
{\row3  \cnz  \cnz  \cnz  \cns  \cns  \cns  \cnz  \cnz};
{\row2  \cnz  \cnz  \cnz  \cns  \cns  \cns  \cnz  \cnz};
{\row1  \cnz  \cnz  \cnz  \cns  \cns  \cns  \cnz  \cnz};
{\row0  \cnz  \cnz  \cnz  \cns  \cns  \cns  \cnz  \cnz};
\end{xy}
\quad
\begin{xy} (0,0);<3mm,0mm>:<0mm,3mm>::
{\grid{b}};
{\row7  \cna  \cna  \cna  \cna  \cna  \cna  \cna  \cna };
{\row6  \cna  \cna  \cna  \cna  \cna  \cna  \cna  \cna };
{\row5  \cna  \cna  \cna  \cna  \cna  \cna  \cna  \cna };
{\row4  \cbbo \cbbo \cbbo \cbbo \cbbo \cbbo \cbbo \cbbo};
{\row3  \cns  \cns  \cns  \cns  \cns  \cns  \cns  \cns };
{\row2  \cns  \cns  \cns  \cns  \cns  \cns  \cns  \cns };
{\row1  \cns  \cns  \cns  \cns  \cns  \cns  \cns  \cns };
{\row0  \cns  \cns  \cns  \cns  \cns  \cns  \cns  \cns };
\end{xy}
\quad
\begin{xy} (0,0);<3mm,0mm>:<0mm,3mm>::
{\grid{c}};
{\row7  \cnb  \cnb  \cna  \cna  \cna  \cna  \cna  \cna };
{\row6  \cnb  \cnb  \cna  \cna  \cna  \cna  \cna  \cna };
{\row5  \cnb  \cnb  \cna  \cna  \cna  \cna  \cna  \cna };
{\row4  \cbha \cbha \cbb  \cbb  \cbb  \cbb  \cbb  \cbb };
{\row3  \cna  \cna  \cns  \cns  \cns  \cns  \cns  \cns };
{\row2  \cna  \cna  \cns  \cns  \cns  \cns  \cns  \cns };
{\row1  \cnz  \cnz  \cns  \cns  \cns  \cns  \cns  \cns };
{\row0  \cnz  \cnz  \cns  \cns  \cns  \cns  \cns  \cns };
\end{xy}
\quad
\begin{xy} (0,0);<3mm,0mm>:<0mm,3mm>::
{\grid{d}};
{\row7  \cnb  \cnb  \cnb  \cnb  \cnb  \cnb  \cna  \cna };
{\row6  \cnb  \cnb  \cnb  \cnb  \cnb  \cnb  \cna  \cna };
{\row5  \cnb  \cnb  \cnb  \cnb  \cnb  \cnb  \cna  \cna };
{\row4  \cbha \cbha \cbha \cbha \cbha \cbha \cbb  \cbb };
{\row3  \cna  \cna  \cna  \cna  \cna  \cna  \cns  \cns };
{\row2  \cna  \cna  \cna  \cna  \cnz  \cnz  \cns  \cns };
{\row1  \cnz  \cnz  \cnz  \cnz  \cnz  \cnz  \cns  \cns };
{\row0  \cnz  \cnz  \cnz  \cnz  \cnz  \cnz  \cns  \cns };
\end{xy}
\]
\caption{After lines \ref{line:ITFT-first-cols-1}--\ref{line:ITFT-first-cols-2} of $\cfITFT$. The bold row is about to be transformed by line \ref{line:ITFT-last-row}.}
\label{fig:ITFT-after-first-cols}
\end{figure}
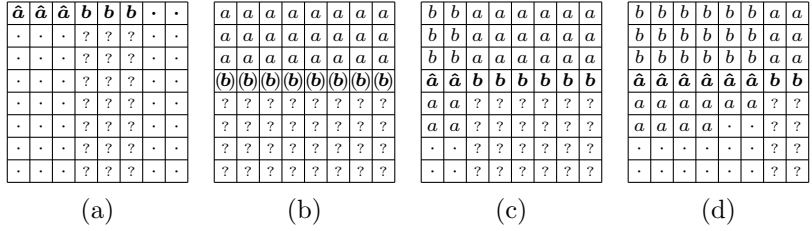

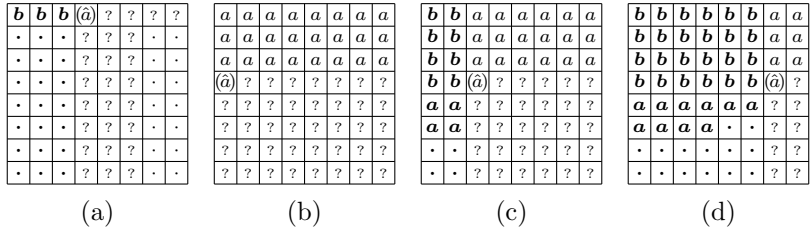
\begin{figure}
\[ 
\begin{xy} (0,0);<3mm,0mm>:<0mm,3mm>::
{\grid{a}};
{\row7  \cbb  \cbb  \cbb  \cnhao\cns  \cns  \cns  \cns};
{\row6  \cbz  \cbz  \cbz  \cns  \cns  \cns  \cnz  \cnz};
{\row5  \cbz  \cbz  \cbz  \cns  \cns  \cns  \cnz  \cnz};
{\row4  \cbz  \cbz  \cbz  \cns  \cns  \cns  \cnz  \cnz};
{\row3  \cbz  \cbz  \cbz  \cns  \cns  \cns  \cnz  \cnz};
{\row2  \cbz  \cbz  \cbz  \cns  \cns  \cns  \cnz  \cnz};
{\row1  \cbz  \cbz  \cbz  \cns  \cns  \cns  \cnz  \cnz};
{\row0  \cbz  \cbz  \cbz  \cns  \cns  \cns  \cnz  \cnz};
\end{xy}
\quad
\begin{xy} (0,0);<3mm,0mm>:<0mm,3mm>::
{\grid{b}};
{\row7  \cna  \cna  \cna  \cna  \cna  \cna  \cna  \cna };
{\row6  \cna  \cna  \cna  \cna  \cna  \cna  \cna  \cna };
{\row5  \cna  \cna  \cna  \cna  \cna  \cna  \cna  \cna };
{\row4  \cnhao\cns  \cns  \cns  \cns  \cns  \cns  \cns  };
{\row3  \cns  \cns  \cns  \cns  \cns  \cns  \cns  \cns };
{\row2  \cns  \cns  \cns  \cns  \cns  \cns  \cns  \cns };
{\row1  \cns  \cns  \cns  \cns  \cns  \cns  \cns  \cns };
{\row0  \cns  \cns  \cns  \cns  \cns  \cns  \cns  \cns };
\end{xy}
\quad
\begin{xy} (0,0);<3mm,0mm>:<0mm,3mm>::
{\grid{c}};
{\row7  \cbb  \cbb  \cna  \cna  \cna  \cna  \cna  \cna };
{\row6  \cbb  \cbb  \cna  \cna  \cna  \cna  \cna  \cna };
{\row5  \cbb  \cbb  \cna  \cna  \cna  \cna  \cna  \cna };
{\row4  \cbb  \cbb  \cnhao\cns  \cns  \cns  \cns  \cns };
{\row3  \cba  \cba  \cns  \cns  \cns  \cns  \cns  \cns };
{\row2  \cba  \cba  \cns  \cns  \cns  \cns  \cns  \cns };
{\row1  \cbz  \cbz  \cns  \cns  \cns  \cns  \cns  \cns };
{\row0  \cbz  \cbz  \cns  \cns  \cns  \cns  \cns  \cns };
\end{xy}
\quad
\begin{xy} (0,0);<3mm,0mm>:<0mm,3mm>::
{\grid{d}};
{\row7  \cbb  \cbb  \cbb  \cbb  \cbb  \cbb  \cna  \cna };
{\row6  \cbb  \cbb  \cbb  \cbb  \cbb  \cbb  \cna  \cna };
{\row5  \cbb  \cbb  \cbb  \cbb  \cbb  \cbb  \cna  \cna };
{\row4  \cbb  \cbb  \cbb  \cbb  \cbb  \cbb  \cnhao\cns };
{\row3  \cba  \cba  \cba  \cba  \cba  \cba  \cns  \cns };
{\row2  \cba  \cba  \cba  \cba  \cbz  \cbz  \cns  \cns };
{\row1  \cbz  \cbz  \cbz  \cbz  \cbz  \cbz  \cns  \cns };
{\row0  \cbz  \cbz  \cbz  \cbz  \cbz  \cbz  \cns  \cns };
\end{xy}
\]
\caption{After line \ref{line:ITFT-last-row} of $\cfITFT$. The bold columns are about to be transformed by lines \ref{line:ITFT-last-cols-1}--\ref{line:ITFT-last-cols-2}.}
\label{fig:ITFT-after-last-row}
\end{figure}

\begin{figure}
\[ 
\begin{xy} (0,0);<3mm,0mm>:<0mm,3mm>::
{\grid{a}};
{\row7  \cna  \cna  \cna  \cnhao\cns  \cns  \cns  \cns};
{\row6  \cns  \cns  \cns  \cns  \cns  \cns  \cnz  \cnz};
{\row5  \cns  \cns  \cns  \cns  \cns  \cns  \cnz  \cnz};
{\row4  \cns  \cns  \cns  \cns  \cns  \cns  \cnz  \cnz};
{\row3  \cns  \cns  \cns  \cns  \cns  \cns  \cnz  \cnz};
{\row2  \cns  \cns  \cns  \cns  \cns  \cns  \cnz  \cnz};
{\row1  \cns  \cns  \cns  \cns  \cns  \cns  \cnz  \cnz};
{\row0  \cns  \cns  \cns  \cns  \cns  \cns  \cnz  \cnz};
\end{xy}
\quad
\begin{xy} (0,0);<3mm,0mm>:<0mm,3mm>::
{\grid{b}};
{\row7  \cna  \cna  \cna  \cna  \cna  \cna  \cna  \cna };
{\row6  \cna  \cna  \cna  \cna  \cna  \cna  \cna  \cna };
{\row5  \cna  \cna  \cna  \cna  \cna  \cna  \cna  \cna };
{\row4  \cnhao\cns  \cns  \cns  \cns  \cns  \cns  \cns  };
{\row3  \cns  \cns  \cns  \cns  \cns  \cns  \cns  \cns };
{\row2  \cns  \cns  \cns  \cns  \cns  \cns  \cns  \cns };
{\row1  \cns  \cns  \cns  \cns  \cns  \cns  \cns  \cns };
{\row0  \cns  \cns  \cns  \cns  \cns  \cns  \cns  \cns };
\end{xy}
\quad
\begin{xy} (0,0);<3mm,0mm>:<0mm,3mm>::
{\grid{c}};
{\row7  \cna  \cna  \cna  \cna  \cna  \cna  \cna  \cna };
{\row6  \cna  \cna  \cna  \cna  \cna  \cna  \cna  \cna };
{\row5  \cna  \cna  \cna  \cna  \cna  \cna  \cna  \cna };
{\row4  \cna  \cna  \cnhao\cns  \cns  \cns  \cns  \cns };
{\row3  \cns  \cns  \cns  \cns  \cns  \cns  \cns  \cns };
{\row2  \cns  \cns  \cns  \cns  \cns  \cns  \cns  \cns };
{\row1  \cns  \cns  \cns  \cns  \cns  \cns  \cns  \cns };
{\row0  \cns  \cns  \cns  \cns  \cns  \cns  \cns  \cns };
\end{xy}
\quad
\begin{xy} (0,0);<3mm,0mm>:<0mm,3mm>::
{\grid{d}};
{\row7  \cna  \cna  \cna  \cna  \cna  \cna  \cna  \cna };
{\row6  \cna  \cna  \cna  \cna  \cna  \cna  \cna  \cna };
{\row5  \cna  \cna  \cna  \cna  \cna  \cna  \cna  \cna };
{\row4  \cna  \cna  \cna  \cna  \cna  \cna  \cnhao\cns };
{\row3  \cns  \cns  \cns  \cns  \cns  \cns  \cns  \cns };
{\row2  \cns  \cns  \cns  \cns  \cns  \cns  \cns  \cns };
{\row1  \cns  \cns  \cns  \cns  \cns  \cns  \cns  \cns };
{\row0  \cns  \cns  \cns  \cns  \cns  \cns  \cns  \cns };
\end{xy}
\]
\caption{After lines \ref{line:ITFT-last-cols-1}--\ref{line:ITFT-last-cols-2} of $\cfITFT$.}
\label{fig:ITFT-done}
\end{figure}

\textbf{Case (a):} $z_1 = 0$. This implies that $0 < n_2 \leq z_2 = z_2' < L_2$, $n_1 = 0$, $m = n_2$, $m' = z_2$, and $f' = 1$. Line \ref{line:ITFT-first-rows} has no effect since $n_1 = 0$. Line \ref{line:ITFT-first-cols-1} computes $x_i = L_2 b_i$ for $n_2 \leq i < z_2$, and destroys $x_i$ for $n_2 \leq i_2 < z_2$, $1 \leq i_1 < L_1$. Line \ref{line:ITFT-first-cols-2} has no effect since $z_2 = z_2'$. Line \ref{line:ITFT-last-row} computes $x_i = L_2 b_i$ for $0 \leq i < n_2$, computes $x_{n_2} = x_n = \hat a_n$ if $f = 1$, and destroys $x_i$ for $n_2 + f \leq i < L_2$. Line \ref{line:ITFT-last-cols-1} computes $x_i = La_i$ for $0 \leq i < n_2 = n$, and destroys $x_i$ for $0 \leq i_2 < n_2$, $1 \leq i_1 < L_1$. Line \ref{line:ITFT-last-cols-2} has no effect since $m = n_2$.

\textbf{Case (b):} $z_1 > 0$ and $n_2 = 0$. This implies that $z_1 \geq n_1 > 0$, $z_2' = L_2$, $m = 0$, $m' = z_2$ and $f' = f$. Line \ref{line:ITFT-first-rows} computes $x_i = L_2 b_i$ for $0 \leq i < n_1 L_2 = n$. Lines \ref{line:ITFT-first-cols-1}--\ref{line:ITFT-first-cols-2} compute $x_i = La_i$ for $0 \leq i < n_1 L_2 = n$, and if $f = 1$ also compute $x_i = L_2 b_i$ for $0 \leq i_2 < L_2$, $i_1 = n_1$; they destroy $x_i$ for $L_2(n_1 + f) \leq i < L$. If $f = 1$, then line \ref{line:ITFT-last-row} computes $x_{n_1 L_2} = x_n = \hat a_n$ and destroys $x_i$ for $n_1 L_2 < i < (n_1 + 1)L_2$. Lines \ref{line:ITFT-last-cols-1}--\ref{line:ITFT-last-cols-2} have no effect since $m = n_2 = 0$.

\textbf{Case (c):} $z_1 > 0$, $n_2 > 0$ and $n_2 \leq z_2$. This implies that $z_2' = L_2$, $0 \leq n_1 < L_1$, $m = n_2$, $m' = z_2$, and $f' = 1$. Line \ref{line:ITFT-first-rows} computes $x_i = L_2 b_i$ for $0 \leq i < n_1 L_2$. For each $n_2 \leq i_2 < L_2$, lines \ref{line:ITFT-first-cols-1}--\ref{line:ITFT-first-cols-2} compute $x_i = L a_i$ for $0 \leq i_1 < n_1$, compute $x_i = L_2 b_i$ for $i_1 = n_1$, and destroy $x_i$ for $n_1 < i_1 < L_1$. Line \ref{line:ITFT-last-row} computes $x_i = b_i$ for $0 \leq i_2 < n_2$, $i_1 = n_1$, computes $x_n = \hat a_n$ if $f = 1$, and destroys $x_i$ for $n_2 + f \leq i_2 < L_2$, $i_1 = n_1$. Finally, for each $0 \leq i_2 < n_2$, lines \ref{line:ITFT-last-cols-1}--\ref{line:ITFT-last-cols-2} compute $x_i = L a_i$ for $0 \leq i_1 < n_1 + 1$ and destroy $x_i$ for $n_1 + 1 \leq i_1 < L_1$.

\textbf{Case (d):} $z_1 > 0$, $n_2 > 0$ and $n_2 > z_2$. The discussion for this case is essentially the same as for (c), with $m$ and $m'$ exchanged, and with slightly different diagrams.

Now we verify the complexity bound. The argument is similar to that used for the TFT. For $L = 2$ the bound is $\min((n+f-1)/2 + 1, 1) = 1$, so the estimate holds. Now assume that $L \geq 4$, and let $\ell_1 = \log_2 L_1$ and $\ell_2 = \log_2 L_2$.

We first verify that the number of calls to the base case is bounded by $L\ell/2$. By induction, lines \ref{line:ITFT-first-cols-1}--\ref{line:ITFT-first-cols-2} and \ref{line:ITFT-last-cols-1}--\ref{line:ITFT-last-cols-2} call the base case at most $L_2(L_1 \ell_1/2)$ times altogether. Lines \ref{line:ITFT-first-rows} and \ref{line:ITFT-last-row} call it at most $L_1(L_2 \ell_2/2)$ times (note that if line \ref{line:ITFT-last-row} is executed then $n_1 \leq L_1 - 1$). The sum is $L_1 L_2 (\ell_1 + \ell_2)/2 = L\ell/2$.

Second, we must verify that the number of calls is bounded by $(n+f-1)\ell/2 + L - 1$.  Line \ref{line:ITFT-first-rows} calls the base case at most $n_1(L_2\ell_2/2)$ times, lines \ref{line:ITFT-first-cols-1}--\ref{line:ITFT-first-cols-2} call it at most $(L_2 - n_2)((n_1 + f' - 1)\ell_1/2 + L_1 - 1)$ times, line \ref{line:ITFT-last-row} calls it at most $f'((n_2 + f - 1)\ell_2/2 + L_2 - 1)$ times, and lines \ref{line:ITFT-last-cols-1}--\ref{line:ITFT-last-cols-2} call it at most $n_2(n_1\ell_1/2 + L_1 - 1)$ times. The sum of these terms is $\frac12 X + Y$, where
\[ \begin{aligned}
 X & = n_1 L_2 \ell_2 + L_2(n_1 + f' - 1)\ell_1 - (f'-1)n_2 \ell_1 + f'(n_2 + f - 1)\ell_2 \\
   & = (n - n_2)\ell + (f' - 1)(L_2 - n_2)\ell_1 + f'(n_2 + f - 1)\ell_2 \\
   & = (n + f - 1)\ell + (f' - 1)(L_2 - n_2)\ell_1 + (n_2 + f - 1)(\ell_2 f' - \ell), \\
 Y & = L_2(L_1 - 1) + f'(L_2 - 1) = L - 1 + (f' - 1)(L_2 - 1).
\end{aligned} \]
If $f' = 1$ then $n_2 + f \geq 1$ and the bound follows since $\ell_2 f' - \ell = -\ell_1 \leq 0$. If $f' = 0$ then $n_2 = f = 0$ and the bound follows since $-L_2 \ell_1 + \ell \leq 0$ (as in the proof of Theorem \ref{thm:TFT}).
\end{proof}

\section{Empirical performance and applications}
\label{sec:applications}

\subsection{The Sch\"onhage--Strassen algorithm}

Both the Magma computer algebra system (version 2.14-15, \cite{magma}) and Victor Shoup's NTL library (version 5.4.2, \cite{ntl}) use the Sch\"onhage--Strassen algorithm \cite{schonhage-strassen} for multiplication of dense polynomials in $\ZZ[x]$ when (roughly speaking) the coefficient size of the input polynomials (in bits) is larger than their degree. The algorithm may be sketched as follows. Suppose that $f, g \in \ZZ[x]$, and put $h = fg$. Let $R = \ZZ/(2^{kN/2} + 1)\ZZ$, where we choose $N = 2^n > \deg h$ and $kN/2$ larger than the size of the coefficients of $h$. Multiply the polynomials in $R[x]/(x^N - 1)$, using an FFT with respect to the principal $N$-th root of unity $\omega_{N} = 2^k \in R$, and lift the result back to $\ZZ[x]$. Arithmetic in $R$ is especially efficient owing to the ease of reduction modulo $2^{kN/2} + 1$ and of multiplication by powers of $\omega_N$.

The author, in joint work with William Hart, implemented the Sch\"onhage--Strassen algorithm using the techniques of this paper to improve smoothness and locality. The implementation is part of the \texttt{fmpz\_poly} module in the FLINT library (version 1.0.13, \cite{flint}), which is used as the default back-end for arithmetic in $\ZZ[x]$ in the Sage computer algebra system (version 3.1.1, \cite{sage}).

The following performance measurements were conducted on a 16-core 2.6GHz Opteron server running Ubuntu Linux. This is a 64-bit processor with a 64 KB L1 cache and 1 MB L2 cache. Only a single core was used for the tests. Our own code and NTL were compiled with gcc 4.1.3, and linked with GMP (GNU Multiple Precision Arithmetic Library, \cite{gmp}) version 4.2.3. We also applied an assembly patch of Pierrick Gaudry that improves the performance of GMP on the Opteron. Magma also uses Gaudry's patch, and links statically against GMP.

\begin{figure}
\begin{center}
\includegraphics[width=0.9\textwidth]{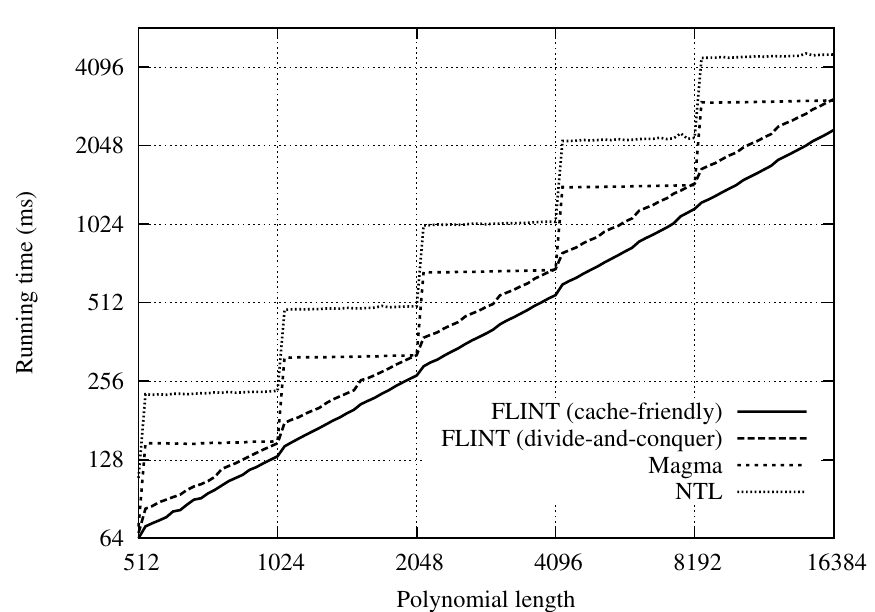}
\caption{Performance of several implementations of the Sch\"onhage--Strassen algorithm for 8000-bit coefficients}
\label{fig:SS-graph}
\end{center}
\end{figure}

Figure \ref{fig:SS-graph} compares four implementations for the case of polynomials with random non-negative $8000$-bit coefficients, with lengths ranging from $512$ to $16384$ in 5\% increments. The graphs for Magma and NTL exhibit the jumps characteristic of FFT-based multiplication algorithms. The two graphs for FLINT show the multiplication performance obtained for van der Hoeven's divide-and-conquer truncated transforms, and for the cache-friendly truncated transforms. The latter is between 15\% and 35\% faster than the former for this range of polynomial lengths, and the relative improvement in performance increases with the degree. Note that the Fourier coefficients are about 16000 bits long ($\approx$ 2 KB), so about 32 coefficients fit into the L1 cache and about 512 coefficients fit into the L2 cache.

\subsection{The Sch\"onhage--Nussbaumer algorithm}

The author implemented the cache-friendy transforms in the context of the Sch\"onhage--Nussbaumer algorithm \cite{schonhage, nussbaumer} for multiplication in $S[x]$ where $S = \ZZ/m\ZZ$ and where $m$ is an odd word-sized modulus. The implementation is part of the \texttt{zn\_poly} polynomial arithmetic library (version 0.9, \cite{zn-poly}). The code has been used in several number-theoretic applications, including computations of zeta functions of hyperelliptic curves over prime fields of large characteristic \cite{kedlaya-large-p}, computations of $L$-functions of hyperelliptic curves over $\QQ$ \cite{L-functions}, computing Hilbert class polynomials \cite{hilbert}, and an ongoing project with Joe Buhler to extend the verification of Vandiver's conjecture and computation of irregular primes and cyclotomic invariants carried out in \cite{buhler}.

The basic idea of the Sch\"onhage--Nussbaumer algorithm is to split the input polynomials into pieces of length $M/4$, and then map the problem to a convolution in $R[z]/(z^K - 1)$ for $R = S[y]/(y^{M/2} + 1)$, where $K \divides M$ so that $R$ contains a principal $K$-th root of unity (namely $y^{M/K}$), and where $K$ is large enough to accommodate the product. Our implementation performs the FFTs over $R$ using the transforms of \S\ref{sec:TFT} and \S\ref{sec:ITFT}, ensuring relatively smooth performance as a function of the input polynomial length. The pointwise multiplications are handled using a multipoint Kronecker substitution method \cite{multipoint}, switching to Nussbaumer's algorithm for sufficiently large $M$. (Note that we do \emph{not} perform an FFT over $\ZZ/m\ZZ$; such an FFT is usually not possible since $\ZZ/m\ZZ$ rarely contains appropriate roots of unity.)

We compared the performance of the cache-friendly transforms to the divide-and-conquer transforms for a range of polynomial lengths ($10^4$ to $3\times 10^7$) and modulus sizes (5 to 63 bits). We observed a modest improvement in speed of up to 15\%, depending on the polynomial length and modulus. As expected, polynomials of higher degree enjoy a greater relative improvement, as locality plays a greater role in such multiplications. Somewhat counterintuitively, the modulus size had the opposite effect on relative performance. This may be explained by noting that the FFTs in our implementation operate on arrays with each element of $\ZZ/m\ZZ$ occupying a single machine word, so the total FFT time does not depend on the modulus; on the other hand, the pointwise multiplications are faster for smaller moduli, as the Kronecker substitution reduces them to smaller integer multiplications. The implementation thus spends a smaller proportion of the total time in the FFTs when the modulus is larger, leading to a smaller relative improvement derived from the cache-friendly transforms.

\section{The small coefficient case}
\label{sec:small}

In the applications described in \S\ref{sec:applications}, elements of the coefficient ring $R$ occupy moderately large blocks of memory. However, FFTs are also commonly applied over `small' coefficients, such as double-precision floating point numbers, or residues modulo a word-sized prime $p$ where $\ZZ/p\ZZ$ contains suitable roots of unity. We have not attempted an implementation in this context, but in this section we make several relevant observations.

An essential consideration in the small coefficient case is spatial locality, which we have largely ignored in this paper. In typical contemporary cache hardware, the cache is organised into cache lines, each capable of storing several words from consecutive locations in main memory. If an algorithm operates on coefficients spaced out in memory, then only a single word of each cache line will be utilised, greatly reducing the effective size of the cache. Moreover, the mapping from physical addresses to cache lines often depends on only the last few bits of the address. If two coefficients are separated by a large power-of-two distance in memory --- exactly the situation during the column transforms of a matrix FFT --- then the cache cannot simultaneously hold both of them (although this can be mitigated to some extent by cache associativity). The standard solution to these problems is to transpose the matrix for the duration of the column transforms, using a cache-friendly matrix transpose algorithm, so that the subtransforms always operate on consecutive data. A similar approach would be needed to adapt our TFTs/ITFTs to the small coefficient case.

A second remark is that in the small coefficient case, it is quite reasonable to zero-pad the inputs so that there is no `partial row'. The rationale is that the lowest level of cache can hold a large number of coefficients, making the penalty for zero-padding quite small. For example, suppose that the cache can hold $2^{13}$ coefficients (typical for a 64KB L1 cache with double-precision floating-point coefficients), and that we are multiplying polynomials whose product has length $n = 12801 = 100 \cdot 2^7 + 1$. This requires a transform length of $2^{14}$, which we may decompose into a $2^7 \times 2^7$ matrix. If we zero-pad the inputs so that $n$ increases to $12928 = 101 \cdot 2^7$, an integral number of rows, the running time penalty incurred is at most 1\%. This approach simplifies the ITFT routine considerably, since it may be implemented by simply reversing the steps of the TFT, removing the need for the special row transform (line \ref{line:ITFT-last-row} of Algorithm \ref{algo:ITFT}). The reduction in code complexity is likely worthwhile. We also note that the presence of a partial row makes it more difficult to maintain spatial locality during the special row transform.

Finally, in the implementations described in \S\ref{sec:applications}, the parameter $\zeta = \omega^s$ is represented simply by the integer $s$. With this representation, computing roots of unity (for example, computing $\zeta^{L_1}$ in line \ref{line:TFT-begin-rows} of Algorithm \ref{algo:TFT}) is very cheap compared to the cost of arithmetic in $R$. In the small coefficient case this is no longer necessarily true, and the cost of computing or storing roots of unity must be taken into account.

\section*{Acknowledgments}

Many thanks to William Hart for his collaboration in implementing these algorithms in FLINT, to William Hart, Andrew Sutherland and Joris van der Hoeven for their comments on a draft of this paper, and to the Department of Mathematics at Harvard University for supplying the hardware on which the performance measurements were carried out.

\bibliographystyle{amsalpha}
\bibliography{cache-trunc-fft}

\end{document}